\documentclass[3p,times,procedia]{elsarticle}
\usepackage{ecrc}
\usepackage[bookmarks=false]{hyperref}
    \hypersetup{colorlinks,
      linkcolor=blue,
      citecolor=blue,
      urlcolor=blue}
\usepackage{amsmath}
\usepackage[Algorithm,ruled]{algorithm}
\usepackage{algorithmicx}
\usepackage[noend]{algpseudocode}
\usepackage{listings}
\usepackage[T1]{fontenc}

\volume{00}

\firstpage{1}

\journalname{Procedia Computer Science}

\runauth{Radu Stefan Mincu}

\jid{procs}

\usepackage{amssymb}
 \usepackage{amsthm}

\usepackage[figuresright]{rotating}
   \newtheorem{Lemma}{Lemma}

\begin{document}
\begin{frontmatter}



\dochead{23rd International Conference on Knowledge-Based and Intelligent Information \& Engineering Systems}%

\title{Heuristic Algorithm for Generalized Function Matching}


\author{Radu Stefan Mincu} 

\address{Department of Computer Science, University of Bucharest, Bucharest, Romania}

\begin{abstract}
The problem of generalized function matching can be defined as follows: given a pattern $p=p_1 \cdots p_m$ and a text $t=t_1 \cdots t_n$, find a mapping $f:\Sigma_p\rightarrow\Sigma_t^{*}$ and all text locations $i$ such that $f(p_1) f(p_2) \cdots f(p_m) = t_i \cdots t_j$, a substring of $t$.

By modifying the restrictions of the matching function $f$, one can obtain different matching problems, many of which have important applications.
When $f:\Sigma_p\rightarrow\Sigma_t$ we are faced with problems found in the well-established field of combinatorial pattern matching. If the single character constraint is lifted and $f:\Sigma_p\rightarrow\Sigma_t^{*}$, we obtain generalized function matching as introduced by Amir and Nor (JDA 2007). If we further constrain $f$ to be injective, then we arrive at generalized parametrized matching as defined by Clifford et al. (SPIRE 2009).

There are a number of important applications for pattern matching in computational biology, text editors and data compression, to name a few. Therefore, many efficient algorithms have been developed for a wide variety of specific problems including finding tandem repeats in DNA sequences, optimizing embedded systems by reusing code etc.

In this work we present a heuristic algorithm illustrating a practical approach to tackling a variant of generalized function matching where $f:\Sigma_p\rightarrow\Sigma_t^{+}$ and demonstrate its performance on human-produced text as well as random strings.
\end{abstract}

\begin{keyword}
string algorithms \sep pattern matching \sep heuristics 



\end{keyword}
\end{frontmatter}

\email{mincu.radu@fmi.unibuc.ro}

\section{Introduction}

Given a pattern $p=p_1 p_2 \cdots p_m$ over an alphabet $\Sigma_p$ and a text $t=t_1 t_2 \cdots t_n$ over an alphabet $\Sigma_t$, the focus of \textit{generalized function matching} (GFM) is to find all the occurrences within the text of the image of the pattern through a function $f$. In other words, the goal of GFM is to find a function $f:\Sigma_p\rightarrow\Sigma_t^{*}$ and all text locations $i$ such that $f(p_1) f(p_2) \cdots f(p_m) = t_i \cdots t_j$ i.e. $f(p)$ is a substring of $t$.

\subsection*{Previous and Related Work}
In \textit{parametrized function matching} introduced by Baker \cite{baker}, the task is to find a bijection $f:\Sigma_p\rightarrow\Sigma_t$. The motivation in Baker's work was related to software maintenance: finding pieces of code that are identical with the exception of variable renaming, with applications in optimizing embedded systems. 

Later on, Amir et al. \cite{amir} introduce \textit{function matching}, a similar problem to parametrized function matching, where $f:\Sigma_p\rightarrow\Sigma_t$ can be any function. The applications include optimizing embedded systems, techniques for computational biology and image processing.

Following that, Amir and Nor \cite{amirnor} propose \textit{generalized function matching} (GFM) where the pattern is used to describe more complicated substructures inside a text, rather than simply 1:1 substring matching. It is here for the first time that the image of the matching function is a substring of the text, rather than a single character i.e. $f:\Sigma_p\rightarrow\Sigma_t^{*}$.

An $\mathcal{NP}$-completeness result is presented for an altered version of the problem, the so-called DGFM or \textit{generalized function matching with don't cares}, where wildcard symbols may appear in the text or pattern. Text wildcards are treated as being equal to all characters of $\Sigma_t$, while pattern wildcards may match any text substring (i.e. the wildcard symbols found in different pattern locations are considered new, distinct symbols of $\Sigma_p$).

Amir and Nor \cite{amirnor} also describe a simple greedy algorithm to solve GFM that is exponential in the size of the pattern alphabet $|\Sigma_p|$. The algorithm is introduced out of theoretical interest as it shows the problem is solvable in polynomial time for bounded alphabets.

In \cite{clifford} Clifford et al. restrict $f$ to be an injective function and introduce \textit{generalized parametrized matching}. Additionally, Clifford et al. produce an $\mathcal{NP}$-completeness proof for GFM when $f:\Sigma_p\rightarrow\Sigma_t^{+}$.

Results regarding the approximability of the problems are given by Clifford and Popa in \cite{cliffpopa}. For a detailed survey regarding generalized function/parametrized matching we refer the reader to Ordyniak and Popa \cite{OrdyniakOther,Ordyniak}.

\begin{Lemma}{Whole text matching is equivalent to substring matching.}
Finding $f:\Sigma_p\rightarrow\Sigma_t^{+}$ such that $f(p) = t_i \cdots t_j$ is equivalent to finding $f:\Sigma_p\cup \{\alpha,\beta\}\rightarrow(\Sigma_t\cup\{\#,\$\})^{+}$ such that $f(p')=f(\alpha p \beta) = \# t \$$, where $\{\alpha,\beta\} \cap \Sigma_p = \emptyset$ and $\{\#,\$\} \cap \Sigma_t = \emptyset$.
\end{Lemma}
\begin{proof}
The newly introduced pattern symbol $\alpha$ can only be matched to the newly introduced text symbol $\#$ or to a non-empty prefix of $t$. Similarly, $\beta$ can only be matched with $\$$ or a non-empty suffix of $t$. The remainder of $t$ (which is a substring) has to match the pattern $p$ in the case of a valid match for $p'$.
\end{proof}

Using the previous result, we may focus on matching the pattern with the entire text instead of every text substring.

We make a note here that for the algorithm to be of practical use, we restrict $f:\Sigma_p\rightarrow\Sigma_t^{*}$ to $f:\Sigma_p\rightarrow\Sigma_t^{+}$. Should we allow empty substrings in the image of $f$, then we can choose to map variations of $p$ from which we have removed some symbol(s) from its alphabet (effectively mapping them to the empty text substring). The practical relevance of this behavior of GFM is arguable, but we choose to disallow empty substrings to be matched hereafter.

The paper is structured as follows. In Section \ref{sec:algo} we give an informal description for the proposed algorithm and afterwards provide a general pseudocode implementation and discuss its running time. Section \ref{sec:decomp} introduces a heuristic for pattern decomposition that helps reduce the running time when the pattern has a certain structure. In Section \ref{sec:exp} we present two experiments, one focusing on English sonnets' structure and the other describing algorithm running time on random data.

\section{The Proposed Algorithm for GFM}
\label{sec:algo}

In this section we detail the inner workings of the heuristic algorithm used to tackle GFM.

\subsection{Description of the Algorithm}

The input to the algorithm is a finite pattern $p=p_1 p_2 \cdots p_m$ and a finite text over some alphabet $t = t_1 t_2 \cdots t_n$. The output of the algorithm is a matching function $f:\Sigma_p\rightarrow\Sigma_t^{+}$ such that $f(p_1) f(p_2) \cdots f(p_m) = t$.

The key idea of the algorithm is to focus on the repetitions in the pattern and on the repetitive structure of the input text. The repetitive structure of the input text can be efficiently discerned by using a suffix tree and this is the advantage we choose to exploit when designing the present algorithm. What remains is to understand how we may match the repetitive structure of the pattern with the repetitive structure of the text.

First, assume that the given pattern has no repetitions. If this is the case then the problem becomes an exercise in listing all possible factorizations into $|p|$ non-empty pieces of all substrings of the text. While not entirely trivial, this is of little practical interest.

Otherwise, there is at least one symbol in the pattern occurring more than once. In this case, we separate the repeating symbols from the non-repeating ones. We will refer to the resulting pair of subsequences as the repetitive subsequence, denoted by $p^{rep}$, and the non-repetitive subsequence, denoted by $p^{nrep}$, respectively (e.g. if \lstinline|$p=$fabbcdaae|, \lstinline|$p^{rep}=$abbaa|, \lstinline|$p^{nrep}=$fcde|).

By using the aforementioned repetitive subsequence $p^{rep}$ we can divide the problem into firstly matching $p^{rep}$ with the text and secondly ensuring that there exists a valid mapping for the rest of the symbols in the non-repetitive subsequence $p^{nrep}$. This will include checking whether the symbols in $p^{nrep}$ are matched to non-empty substrings of the text. 


The overall strategy to tackle GFM can now be formulated as follows:
\begin{enumerate}
\item Discern the repetitive structure of the input text (here we use a suffix tree).
\item Process the pattern, separating a repetitive subsequence and a non-repetitive subsequence.
\item Find the matches of the repetitive subsequence with the text.
\item Ensure that the symbols in the non-repetitive subsequence can be matched with non-empty strings.
\end{enumerate}

\subsection{The Subsequence Matching in Detail}

\subsubsection*{Text Processing}

First of all, it is necessary to build a suffix tree for the input text and store the repetition information into a suitable structure. For the suffix tree construction step we use an implementation of Ukkonen's algorithm \cite{Ukkonen}.  To record a repetition of a substring of the text we have considered the following structure $(length,occ,[i_1,\dots,i_{occ}])$, where $length$ is the number of characters in the repeating substring, followed by a list of positions where the string appears in the text ($i_t$, $ 2\leq t \leq occ$), $occ$ being the number of occurrences. In the rest of the paper, the term \textit{repetition} will be interpreted as this previously defined structure.

A list of repetition tuples as described above may be obtained by traversing the constructed suffix tree. Upon traversal, we may choose to discard some redundant information: the suffixes of a repeating substring that have the same number of occurrences as the substring. The benefits of this choice include reducing the search space and memory required for the algorithm. However, this leads to missing some solutions. To ensure algorithm correctness, the entire list of repetition tuples is required. In our heuristic approach, we may choose to search a reduced repetition list for the first iterations of the algorithm (leading to a much faster search) and then gradually recover the lost information in subsequent iterations, should no match be found.

After the repetition list is sorted by length we merge all of the position lists of the repetitions (and keep a record of the index of the corresponding repetition and the index of the position inside the position list). In case of same starting positions, the position of the shorter repeated string appears first for reasons soon to be discussed in the matching algorithm. This merged list will be used to match the repetition pattern to the text.

\subsubsection*{Pattern Processing}

For the pattern processing step, we first count the occurrences of each pattern symbol. Symbols that appear multiple times will form the repetitive subsequence. Additionally, we record the length of each substring made of non-repeating symbols appearing between each repeating symbol. For an example consider \lstinline|$p=$fabbcdaae|: the repetitive subsequence is \lstinline|$p^{rep}=$abbaa| and the non-repeating string length information can be represented as $count^{nrep}=$ 1:0:0:2:0:1. 

\subsubsection*{Subsequence Matching}

To find a match between the repetitive subsequence and the text, we will map (substitute) each symbol $\sigma$ in the repetitive subsequence of the pattern with the index of a repetition having at least as many occurrences as the occurrences of the symbol in the pattern ($occ \geq |p|_\sigma$).

The key  observation is that if the repetitive pattern matches the text, then, for some repetition index substitution, the substituted pattern will be a subsequence inside the merged repetition list. While this is bound to happen, it may still be difficult to obtain a valid matching $f$.

\subsubsection*{String Trimming}

Once an occurrence of the repetitive subsequence has been detected in the merged repetition list, it is necessary to ``trim'' any overlapping substrings. In short, the starting point of any substring will cut the length of the previous substring. This change will affect all substrings that belong to a repetition (they must be identical). If the trimming operation ends up producing empty strings, then the match is rejected at this point. 


Another detail of the trimming process is that it is necessary to sometimes shorten a repeating string (cut away a suffix) to make room for the non-repetitive pattern substring in-between. In other words, we must allow the length of the text substring between repeating pattern symbol matches to be at least the length of the non-repetitive pattern substring in-between (effectively ensuring that non-repeating pattern symbols can each be matched to a non-empty string). 

After trimming, we use the information regarding the non-repetitive string lengths to enforce adjacency between repeating strings should the length be 0. If this is not possible, reject at this step.

%
%

\subsubsection*{Repetition Splitting}

If the algorithm reports no match (or insufficient matches) after the previous steps then it is necessary to reintroduce some of the redundant information that was discarded during the repetition list construction. To do this, we simply split each repetition into ``halves'' by the following method: \lstinline|$\{$abcd$\}$| $\rightarrow$ \lstinline|$\{$abcd$,$cd$\}$| $\rightarrow$ \lstinline|$\{$abcd$,$bcd$,$cd$,$d$\}$|. Correctness is only ensured by employing the previously described trimming procedure.

By performing this operation multiple times on the repetition list we can recover all of the discarded information. The reason for recovering this information gradually with this heuristic is that the size of the repetition list greatly impacts the expected running time.

\subsection{Algorithm pseudo-code}

For an insight into the intricacy of the algorithm described in this section, we refer the reader to the algorithm blocks that contain pseudocode for our proposed algorithm (Proposed Algorithm block) and for the previous algorithm by Amir and Nor \cite{amirnor}, augmented with some optimizations (Previous Algorithm block).

\begin{algorithm}[h!]
\floatname{proposed}
\caption{\textbf{Proposed Algorithm:} whole-text matches input pattern $p$ with input text $t$ and outputs the respective partitioning of $t$. The notation $|w|_a$ denotes a procedure that returns the number of occurrences of symbol $a$ in string $w$. The operator $\cdot$ denotes string concatenation. $L^X$ is the set of functions $map: X \rightarrow L$. To keep a high-level view, the complete trimming procedure is not included (it is lengthy, contains subcases for the beginning and ending of the text etc.). The variable $subseq$ denotes a list of positions in the text where the respective subsequence occurs.

\textit{input:} string $p$,$t$;
\textit{output:} integer $u[1..|p|$]; (such that $t_{u[i]} \cdots t_{u[i+1]}=f(p_i)$)
}
\label{alg:greedy1}
\hrule
\begin{algorithmic}[1]
\State $t \gets \# t \$$, $p^{rep} \gets \emptyset$, $count^{nrep}_1 \gets 0$, $j \gets 1$
\State build a suffix tree $ST(t)$
\State traverse $ST(t)$ and build a list $L$ of repetition structures $(length,occ,pos[1..occ])$
\State $repStruct \gets$ merge and sort the $pos$ arrays of all repetitions in $L$
\For {$i \gets 1,|p|$}
\If {$|p|_{p_{i}}>1$}
\State $p^{rep} \gets p^{rep} \cdot p_i$
\State $j \gets j+1$, $count^{nrep}_j \gets 0$
\Else
\State $count^{nrep}_j \gets count^{nrep}_j+1$
\EndIf
\EndFor
\State $X = \Sigma_p^{rep} \gets \Sigma(p^{rep})$, 
\For {all mapping functions $map \in L^X$}
\For {all occurrences $subseq$ of $map(p^{rep})$ as a subsequence in $repStruct$ }
\State $L^{copy} \gets L$
\For{$j \gets 1,|subseq|$}
\State trim the $length$ of $L^{copy}_i$ by $subseq_{j+1} - count^{nrep}_j$ where $map(p^{rep}_j)=L^{copy}_i$
\State insert the starting positions of the non-repetitive text substrings into $subseq$
\EndFor
\If {trimming produces no empty substring}
\State \textbf{output} $subseq$
\EndIf 
\EndFor
\EndFor

\end{algorithmic}
\end{algorithm}

\begin{algorithm}[!hbt]
\floatname{previous}
\caption{\textbf{Previous Algorithm:} matches input pattern $p$ with input text $t$ and outputs the respective partitioning of $t$.

\textit{input:} string $p$,$t$;
\textit{output:} integer $u[1..|p|$]; (such that $t_{u[i]} \cdots t_{u[i+1]}=f(p_i)$)
}
\label{alg:greedy1}
\hrule
\begin{algorithmic}[1]
\State $i \gets 1$
\For {$c \in \Sigma_p$}
\State $index_c \gets i$, $lengths_i \gets 1$, $i \gets i+1$
\EndFor
\For {$i \gets 1,|t|$}
\State $availableLen \gets |t|-i$
\For {every combination of $lengths$ s.t. $sum(lengths) \leq availableLen$}
\For {$j \gets 1,|\Sigma_p|$}
\State $starts_j \gets =-1$
\EndFor
\State $OK \gets true$, $last \gets i$
\For {$j \gets 1,|p|$}
\State $c \gets p_j$, $k \gets index_c$
\If {$starts_k=-1$}
\State $starts_k \gets last$
\Else
\For {$\ell \gets last,last+lengths_k$}
\If {$t_{\ell} \neq t_{\ell+start_k-last}$}
\State $OK \gets$ false, \textbf{break}
\EndIf
\EndFor
\If{$OK=$ false}
\State \textbf{break}
\EndIf
\EndIf
\State $last \gets last + lengths_k$
\EndFor
\State $last=i$, \textbf{output} $last$
\For {$j \gets 1,|p|$}
\State $c \gets p_j$, $k \gets index_c$
\State $last \gets last + lengths_k$, \textbf{output} $last$  
\EndFor
\EndFor
\EndFor
\end{algorithmic}
\end{algorithm}

\subsection{Time complexity analysis}

In the upcoming analysis we consider $|t|=n$ and $|p|=m$. Additionally we use the denomination \textit{uniformly random string} or \textit{uniform string} when we refer to a string that is built by concatenating symbols in a finite, given alphabet and each symbol has equal probability to be chosen for each position in the string.

From the proposed algorithm's pseudocode we observe that the computation time is dominated by the two loops in lines 11-12. Before that, a suffix tree is built in linear time $O(n)$ using Ukkonen's algorithm, then the list of repetitions $L$ is recovered by traversing the tree in $O(n)$. Following, the repetition list has $O(n)$ structures each containing $O(n)$ integers in their list of occurrences (in the case of uniformly random strings the expected length of these lists is $O(n|\Sigma_t|^{-1})$). As a result, merging and sorting the lists of occurrences takes up to $O(n^2 \log n)$ time (and $O(n^2 |\Sigma_t|^{-1} \log n^2 |\Sigma_t|^{-1})$ expected time for uniform strings) using comparison based sorting. The expected size of the sorted occurrence list is therefore between $O(n^2 |\Sigma_t|^{-1})$ in the uniform string case and $O(n^2)$ in general . The pattern processing step is completed in $O(m)$.

The dominating term of the time complexity of the proposed algorithm is depending on the size of the repetition list. The number of $map: \Sigma^{rep}_p \rightarrow L$ functions is $O(n^{|\Sigma_p|}$). For each such function, we search for a subsequence of length $|p^{rep}|$ inside the merged repetition occurrence list.

The search for the first occurrence of a subsequence in the lexicographic order can be done in linear time with regard to the size of the occurrence list (uniform string case $O(n^2 |\Sigma_t|^{-1})$, in general $O(n^2)$). The next subsequence ocurrence in the lexicographic order of positions can be found in time $O(mn^2 |\Sigma_t|^{-1})$ for the uniform string case, $O(mn^2)$ in general). In the worst case, there are generally ${n^2 \choose m} = O(n^{2m})$ subsequence occurrences and $O(n^{2m}|\Sigma_t|^{-m})$ for uniform strings. For each occurrence a copy of the repetition list is made. Actually, only the repetitions that appear in the image of the mapping function need to be duplicated, so this step takes only $O(|\Sigma_p|)$.

The trimming of the length of each repetition occurrence against the starting position of the next occurrence will take constant time, so the entire trimming procedure takes time $O(m)$, including the handling of the non-repetitive subsequence.

With the previous explanation, the total time complexity of the algorithm is generally $O(mn^{2m+|\Sigma_p|})$ and for uniform strings $O(mn^{2m+|\Sigma_p|} \cdot |\Sigma_t|^{-m})$.

The reference implementation of the algorithm by Amir and Nor has been optimized to $O(n^{|\Sigma_p|+1})$.

\section{Pattern decomposition heuristics}
\label{sec:decomp}

Our search for matches is based around finding a subsequence within the repetition structure obtained from all possible substitutions of the repetitive pattern subsequence with a valid assignment. As such, even for relatively small tests, if the text has a highly repetitive structure, then the search is prohibitively slow.

Our proposed approach to mitigate this drawback is a divide-and-conquer heuristic. The idea is that if a pattern can be split into pieces that have disjoint alphabets, then we can try to match once of the pieces to obtain a candidate solution fragment to the whole pattern match problem. Subsequently, for each candidate fragment identified, we solve the problem of matching the remaining pieces of the pattern with the remaining pieces of the text.

To illustrate our approach consider the following example. Let $p=ABABCDCDEFEFGG$ the rhyme structure of Shakespearean sonnets (see Subsection \ref{sec:toy} for an experiment regarding this pattern). We can split $p$ in the following format to facilitate our search: $p=ABAB \cdot CDCD \cdot EFEF \cdot GG$. In other words, first we locate an occurrence of $ABAB$, then we find a match of $CDCD$ that begins on the first position of the remaining text, and so on.

This pattern splitting method is better formalized by using whole-text matching. If we try to match $p=*ABABCDCDEFEFGG*$ to the entire text, then we can easily represent the split pattern as $(*ABAB*, CDCD*, EFEF*, GG*)$. The previously described technique only works left-to-right: first $*ABAB*$ is matched to the text, then $CDCD*$ is matched to the text substring that matches the rightmost $*$ symbol from $*ABAB*$, then $EFEF*$ is matched to the text substring matching $*$ from $CDCD*$, and so on.

So far we have described a type of pattern splitting of $p$ into as many pieces $q_i$ as possible such that $p=q_1 \cdots q_k$ and all $q_i$ have disjoint alphabets i.e. $\Sigma(q_i)\cap \Sigma(q_j) = \emptyset$, $\forall i \neq j$. This method is based on splitting the pattern in the manner $p= q r$ with $\Sigma(q)\cap \Sigma(r) = \emptyset$. A more general version of this pattern splitting method is to reduce the pattern ``inwards'' as inspired by palindrome type patterns. This method will separate $p = q p' r$ with $p' \cap ( q \cup r ) = \emptyset$. In this case, a whole-text match of $p = * q p' r *$ will be separated into matching $* q \! * \! r *$ first and subsequently matching $p'$ with the text substring matching the middle wildcard of $* q \! * \! r *$.

\section{Experimental results}
\label{sec:exp}

\subsection{A toy experiment}
\label{sec:toy}

To showcase the proposed algorithm, we prepare a toy example inspired from \cite{amirnor}, where the authors give music and poetry analysis as a motivation for GFM and give the example of the Shakespearean sonnet: a poem with a very rigid structure that respects the rhyme template $ABABCDCDEFEFGG$.

In this first experiment we use as a pattern $p=ABABCDCDEFEFGG$ and as a text we randomly select 6 sonnets from the work of Shakespeare, concatenate them, transcribe them to IPA (International Phonetic Alphabet) and keep only the last character of each verse.

In Figure \ref{fig:sonnets} we use a graphical representation of the text by using a technique similar to the one proposed by Wattenberg \cite{Wattenberg} namely the arc diagram, where the ends of each arc indicate exact occurrences of a substring and the arc width represents the length of the repeated substring. Below the gray repetition structure we also display the matches found with our algorithm using red hues. In blue we show a correct match that is not a sonnet (but has a structure that matches the pattern).

It is a noteworthy and surprising find that the third sonnet is not a match for the given pattern. Upon closer inspection the culprit is Shakespeare's Sonnet XLIX, which uses an imperfect rhyme in verses 5-8. Thus, the mismatch is not through any fault of the algorithm (it is working correctly for proper matches).

\begin{figure}
\caption{Graphical representation of the repetitive structure (in grey) of a string obtained by concatenating the last character in the International Phonetic Alphabet transcription of each verse of 6 concatenated Shakespearean sonnets (delimited by the red markers). Mirrored and in red hues are the matches of the Shakespearean sonnet rhyme pattern $ABABCDCDEFEFGG$ as produced by our proposed algorithm. In blue is an example of a correct pattern match that is not a sonnet itself, but is an effect of concatenation.} \label{fig:sonnets}

\includegraphics[width=\textwidth]{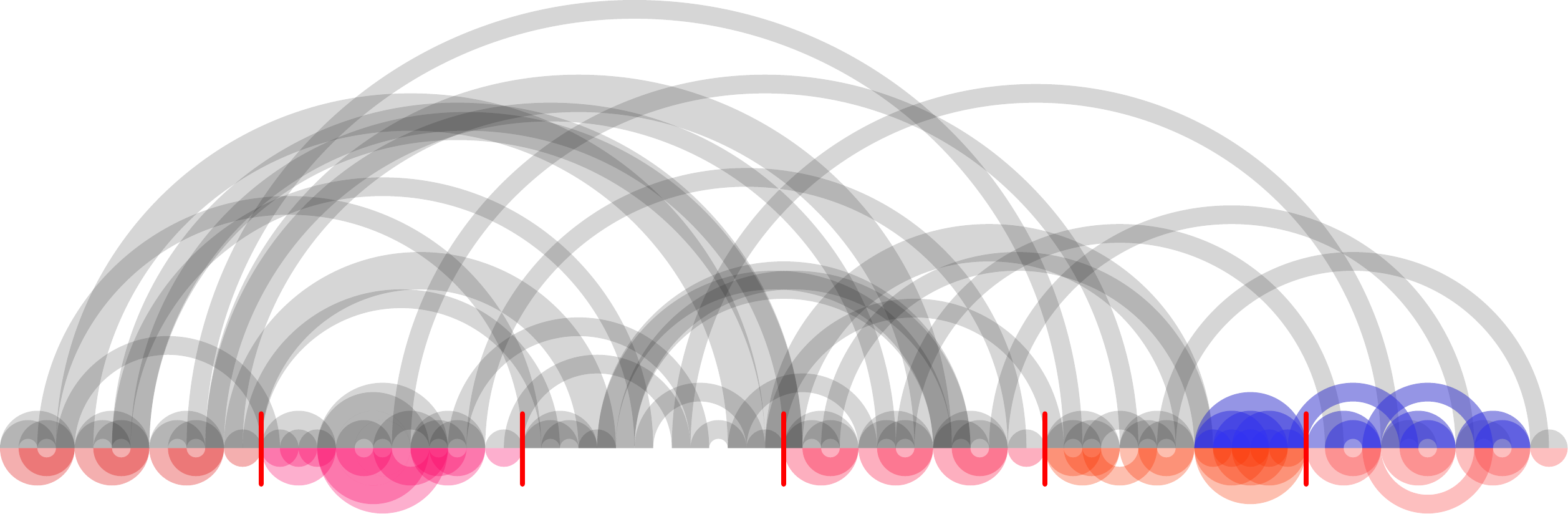}
\end{figure}

\subsection{Running time analysis}

For the purpose of determining running time, our proposed algorithm against is tested against our own optimized implementation of the previous Greedy algorithm proposed by Amir and Nor \cite{amirnor}.

The dataset used in the experiments contains uniformly random strings (meaning that each position in an uniform string can be any symbol in a fixed alphabet $\Sigma$ with equal probability $\frac{1}{|\Sigma|}$). The test patterns have the additional restriction that the first and last symbol must repeat in the pattern (at least once) and there can be no two consecutive non-repeating symbols. If this were not the case then the number of valid solutions can needlessly multiply by $O(|t|)$ e.g. when $p=ab$ is matched to any text $t$ then $a$ is matched to $|t|-1$ prefixes of the text. 

We summarize the results of the experiments using Figures \ref{fig:plot1} and \ref{fig:plot2}: on the $x$-axis are separate test instances, while on the $y$-axis we plot the recorded wall clock time. Both algorithms are run on every test instance and we do not employ pattern decomposition heuristics (from Section \ref{sec:decomp}) for a proper comparison. The results in Figures \ref{fig:plot1} and \ref{fig:plot2} are for 500 instances each.

In Figure \ref{fig:plot1} we observe the exponential growth in time required for the proposed algorithm with regard to the amount of distinct repeating substrings in the text (which grows inversely proportional to the size of the text alphabet). We also showcase the differences induced by different patterns, which end up being quite significant. Here we analyze the algorithms on all patterns of length 5 and alphabet size 3 that respect the previously mentioned restriction (amounting to 10, disregarding isomorphism). The previous algorithm by Amir and Nor seems mostly unaffected by text repetition, showing only a weak linear growth with regard to this parameter (which is in the range 60-130).

In Figure \ref{fig:plot2} we show the exponential growth in the time of the algorithm by Amir and Nor with regard to pattern alphabet size. All of the test instances use the same text and the pattern alphabet varies. We sort the results by pattern alphabet size and execution time of the previous algorithm. While the greedy algorithm is quite predictable and is highly dependent on pattern alphabet, there is much greater variance in the time required for the proposed algorithm, which is highly dependent on the pattern structure as well (adjacency of repeating pattern symbols etc.). For all of the instances, we stop the algorithms when the time exceeds 1 second. It is noteworthy that the expected running time for the algorithm of Amir and Nor for pattern alphabet size 4 is around 21-25 seconds.

It is interesting that if we keep the best time of the two algorithms, the test instances are all solved under 1 second.

\begin{figure}
\caption{Algorithm running time for uniformly random pattern strings and uniformly random text strings with the following parameters: fixed pattern length $|p|=5$, fixed pattern alphabet size $|\Sigma_p|=3$, fixed text length $|t|=200$, variable text alphabet size $|\Sigma_t| \in \{4,\dots,16\}$.} \label{fig:plot1}
\includegraphics[width=\textwidth ,trim={0 0 0.53cm 0},clip]{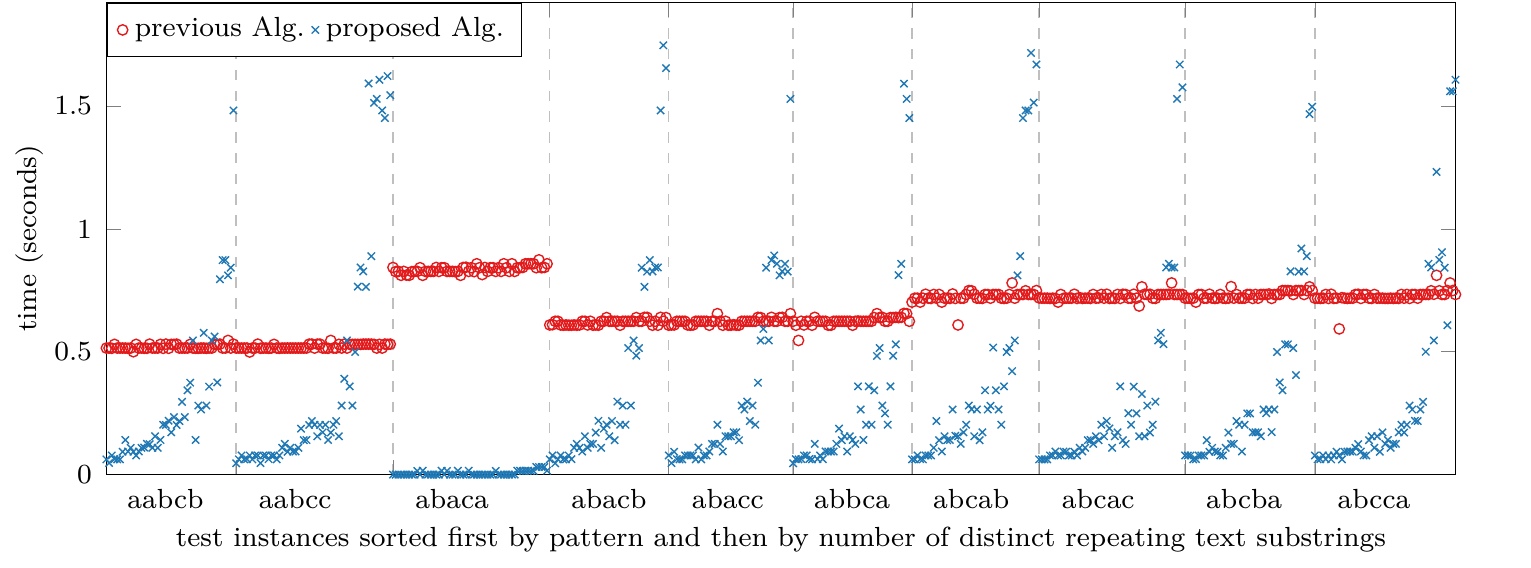}
\end{figure}

\begin{figure}
\caption{Algorithm running time for uniformly random pattern strings and only one uniformly random text string with the following parameters: fixed text (fixed text length $|t|=200$, fixed text alphabet size $|\Sigma_t|=16$, distinct repeating substring count is 70), fixed pattern length $|p|=6$, variable pattern alphabet size $|\Sigma_p| \in \{2,3,4\}$.} \label{fig:plot2}
\includegraphics[width=\textwidth]{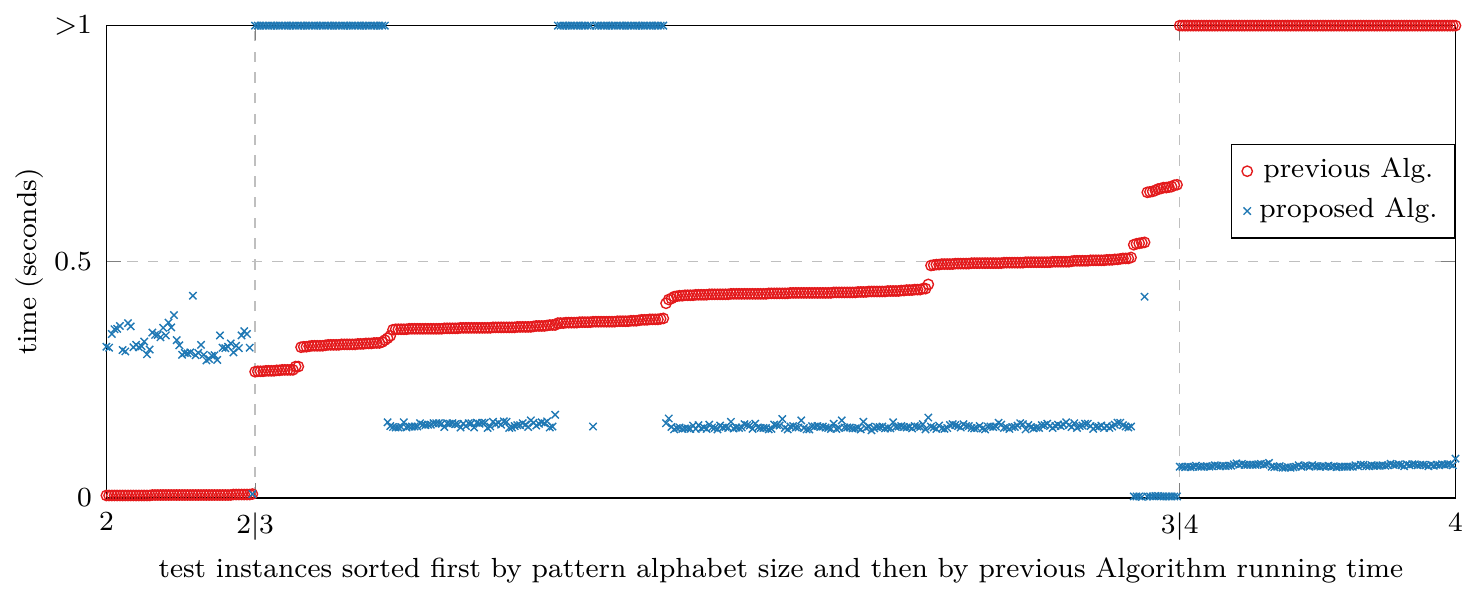}
\end{figure}

\section{Conclusion and open problems}

We show that while the algorithm by Amir and Nor works well for patterns with a small alphabet, there are cases where our algorithm is better, notably on text that has a low number of repeating substrings (which is commonly the case for English text, since the alphabet is not small enough to prohibit our proposed algorithm). We also notice that there are types of patterns for which our algorithm outperforms the previously known algorithm.

For future work it is of interest to design more algorithms for generalized function matching and develop a specialized toolkit that takes into account the pattern structure and the text repetitiveness and automates the choice of the appropriate algorithm.

\bibliographystyle{elsarticle-harv}
\bibliography{bibliography}

\end{document}